\documentclass[a4paper,11pt,UKenglish,dvipsnames]{article}

\usepackage{multirow}

\usepackage{xcolor}
\usepackage{url} 
\usepackage{amssymb,amsbsy}
\usepackage{amsmath} 
\usepackage{amsthm}
\usepackage[pagebackref]{hyperref}  
\usepackage{algorithmic}
\usepackage{algorithm} 
\usepackage{complexity} 
\usepackage{graphicx}
\usepackage{enumerate} 
\usepackage{tikz-cd}

\newtheorem{theorem}{Theorem}
\newtheorem{theorem*}{Theorem}

\newtheorem{lemma}[theorem]{Lemma}

\newcommand{\claimproof}[2]%
{\noindent{\em Proof of Claim \ref{#1}.} #2\hspace*{\fill}$\Box$~~~~~\vspace{5mm} }

\newtheorem{problem}[]{Problem}
\newtheorem{exmp}{Example}
{\vspace{1ex}\noindent{\em Pf.}\hspace{0.5em}\def\myproof@name{#1}}%
{\hfill{\tiny \qed\ (\myproof@name)}\vspace{1ex}}

\newenvironment{proof-sketch}{\medskip\noindent{\em Sketch of
Proof.}\hspace*{1em}}{\qed\bigskip}
\newenvironment{proof-attempt}{\medskip\noindent{\em Proof attempt.}\hspace*{1em}}{\bigskip}

\newcommand{\inbrace }[1]{\left\{#1\right\}}           

\newcommand{\F}{\mathbb{F}}
\newcommand{\N}{\mathbb{N}}

\newcommand{\calH}{\mathcal{H}}

\newcommand{\calP}{\mathcal{P}}
\newcommand{\aff}{\mathbb{A}}
\newcommand{\proj}{\mathbb{P}}

\ifx \@R \@empty

\else

\fi

\ifx \@C \@empty
\newcommand{\C}{\mathbb{C}}
\else
\renewcommand{\C}{\mathbb{C}}
\fi

\renewcommand{\epsilon}{\varepsilon}
\newcommand{\pder}[2]{\partial_{#2}#1}

\newcommand{\ol}[1]{\overline{#1}}

\renewcommand{\dim}{\text{dim~}}
\newcommand{\rk}{\text{rank}} 
\newcommand{\sz}{\text{size}}

\newcommand{\im}{\mathrm{Im}}
\newcommand{\zer}{\mathrm{Zer}}

\newcommand{\m}{\mathfrak{m}}

\begin{document}

\title{\bf Algebraic dependencies and PSPACE algorithms in approximative complexity}

\author{
Zeyu Guo \thanks{Department of Computer Science \& Engineering, Indian Institute of Technology Kanpur, \texttt{zguo@cse.iitk.ac.in} }
\and
Nitin Saxena \thanks{CSE, IIT Kanpur, \texttt{nitin@cse.iitk.ac.in} }
\and
Amit Sinhababu \thanks{CSE, IIT Kanpur, \texttt{amitks@cse.iitk.ac.in} }
}


\date{}
\maketitle

%


\vspace{-.05in}
\begin{abstract}

  Testing whether a set $\mathbf{f}$ of polynomials has an algebraic dependence is a basic problem with several applications. The polynomials are given as algebraic circuits. Algebraic independence testing question is wide open over finite fields (Dvir, Gabizon, Wigderson, FOCS'07). The best complexity known is NP$^{\#\P}$ (Mittmann, Saxena, Scheiblechner, Trans.AMS'14). In this work we put the problem in AM $\cap$ coAM. In particular, dependence testing is unlikely to be NP-hard and joins the league of problems of ``intermediate''  complexity, eg.~graph isomorphism \& integer factoring. Our proof method is algebro-geometric-- estimating the size of the image/preimage of the polynomial map $\mathbf{f}$ over the finite field. A {\em gap} in this size is utilized in the AM protocols.

Next, we study the open question of testing whether every annihilator of $\mathbf{f}$ has zero constant term (Kayal, CCC'09). We give a geometric characterization using Zariski closure of the image of $\mathbf{f}$; introducing a new problem called {\em approximate} polynomials satisfiability (APS). We show that APS is NP-hard and, using projective algebraic-geometry ideas, we put APS in PSPACE (prior best was EXPSPACE via Gr\"obner basis computation). As an unexpected application of this to approximative complexity theory we get-- Over {\em any} field, hitting-set for $\overline{\rm VP}$ can be designed in PSPACE. This solves an open problem posed in (Mulmuley, FOCS'12, J.AMS 2017); greatly mitigating the GCT Chasm (exponentially in terms of space complexity).
\end{abstract}

\vspace{-.30mm}
\noindent
{\bf 1998 ACM Classification:} I.1 Symbolic and Algebraic Manipulation, F.2.1 Numerical Algorithms and Problems, F.1.3 Complexity Measures and Classes, G.1.2 Approximation.

\vspace{-.45mm}
\noindent
{\bf Keywords:} algebraic dependence, Jacobian, Arthur-Merlin, approximate polynomial, satisfiability, hitting-set, border VP, finite field, PSPACE, EXPSPACE, GCT Chasm.

\vspace{-1mm}

\section{Introduction}
\vspace{-0.05in}

Algebraic dependence is a generalization of linear dependence. Polynomials $f_1,\ldots ,f_m \in \F[x_1, \ldots ,x_n]$ are called \emph{algebraically dependent} over field $\F$ if there exists a nonzero polynomial (called \emph{annihilator}) $A(y_1,\ldots ,y_m)\in \F[y_1,\ldots ,y_m]$ such that $A(f_1,\ldots,f_m)=0$. If no $A$ exists, then the given polynomials are called \emph{algebraically independent} over $\F$. The \emph{transcendence degree} (trdeg) of a set of polynomials is the analog of rank in linear algebra. It is defined as the maximal number of algebraically independent polynomials in the set. Both algebraic dependence and linear dependence share combinatorial properties of the matroid structure \cite {ehrenborg1993apolarity}. The algebraic matroid examples may not be linear (esp.~over $\F_p$) \cite{ingleton1971representation}.

The simplest examples of algebraically independent polynomials are $x_1,\ldots,x_n \in \F[x_1, \ldots$, $x_n] $. As an example of algebraically dependent polynomials, we can take $f_1=x$, $f_2=y$ and $f_3=x^2+y^2$. Then, $y_1^2 + y_2^2-y_3$ is an annihilator. The underlying field is crucial in this concept. For example,  polynomials $x+y$ and $x^p + y^p$ are algebraically dependent over $\mathbb{F}_p$, but algebraically independent over $\mathbb{Q}$.  

Thus, the following computational question {\em AD($\F$)} is natural and it is the first problem we consider in this paper: 
Given algebraic circuits $f_1,\ldots ,f_m \in \F[x_1, \ldots ,x_n]$, test if they are algebraically dependent. 
It can be solved in PSPACE using a classical result due to Perron \cite{Per27,Plo05,C76}. Perron proved that given a set of algebraically dependent polynomials, there exists an annihilator whose degree is upper bounded by the product of the degrees of the polynomials in the set. This exponential degree bound on the annihilator is tight \cite{Kay09}. 
 
The annihilator may be quite hard, but it turns out that the decision version is easy over zero (or large) characteristic using a classical result known as the Jacobian criterion \cite{Jac41, BMS13}. The Jacobian efficiently reduces algebraic dependence testing of $f_1,\ldots,f_m$ over $\F$ to linear dependence testing of the differentials $df_1,\ldots,df_m$ over $\F(x_1,\ldots,x_n)$, where we view $df_i$ as the vector $( \frac{\partial f_i} {\partial x_1},\ldots, \frac{\partial f_i} {\partial x_n})$. Placing $df_i$ as the $i$-th row gives us the Jacobian matrix $J$ of $f_1,\ldots,f_m$. If the characteristic of the field is zero (or larger than the product of the degrees $\deg(f_i)$) then the trdeg equals $\rk(J)$.
It follows from \cite{Sch80} that, with high probability, $\rk(J)$ is equal to the rank of $J$ evaluated at a random point in $\F^n$. This gives a simple randomized polynomial time algorithm solving AD($\F$) for certain $\F$.
 
For fields of positive characteristic, if the polynomials are algebraically dependent, then their Jacobian matrix is not full rank. But the converse is not true. There are infinitely many input instances (set of algebraically independent polynomials) for which Jacobian fails. The failure can be characterized by the notion of `inseparable extension' \cite{pandey2016algebraic}. For example, $x^p, y^p$ are algebraically independent over $\mathbb{F}_p$, yet their Jacobian determinant vanishes. Another example is, $\{ x^{p-1}y, xy^{p-1} \}$ over $\mathbb{F}_p$ for prime $p>2$. 
\cite{mittmann2014algebraic} gave a criterion, called Witt-Jacobian, that works over fields of prime characteristic $p$; improving the complexity of independence testing problem from PSPACE to NP$^{\#\P}$. \cite{pandey2016algebraic} gave another generalization of Jacobian criterion that is efficient in special cases.
 
Given that an efficient algorithm to tackle prime characteristic is not in close sight, one could speculate the problem to be NP-hard or even outside the polynomial hierarchy PH. In this work we show that: {\em For finite fields, AD($\F$) is in AM $\cap$ coAM} (Theorem \ref{thm_amcoam}). This rules out the possibility of NP-hardness, under standard complexity theory assumptions \cite{AB09}.
 
 \smallskip\noindent
\textbf{Constant term of the annihilators.}
We come to the second problem {\em AnnAtZero} that we discuss in this paper: Testing if the constant term of {\em every} annihilator, of the set of algebraic circuits $\mathbf{f}=\{f_1,\ldots,f_m\}$, is zero. 
 Note that the annihilators of $\mathbf{f}$ constitute an ideal of the polynomial ring $\F[y_1,\ldots,y_m]$; this ideal is principal when trdeg of $\mathbf{f}$ is $m-1$ \cite[Lem.7]{Kay09}. In this case, we can decide if the constant term of the minimal annihilator is zero in PSPACE, as the {\em unique} annihilator (up to scaling) can be computed in PSPACE. 
 
 If trdeg of $\mathbf{f}$ is less than $m-1$, the ideal of the annihilators of $\mathbf{f}$ is no longer principal. Although the ideal is finitely generated, finding the generators of this ideal is computationally very hard. (Eg.~using Gr\"obner basis techniques, we can do it in EXPSPACE \cite[Sec.1.2.1]{derksen2015computational}.) In this case, can we decide if all the annihilators of $\mathbf{f}$ have constant term zero? {\em We give two equivalent characterizations of AnnAtZero-- one geometric and the other algebraic --using which we devise a PSPACE algorithm to solve it in all cases} (Theorem \ref{thm-aps}).
 
Interestingly, there is an algebraic-complexity application of the above algorithm. {\em We give a PSPACE-explicit construction of a hitting-set of the class $\overline{\rm VP}_{\ol{\F}_q}$} (Theorem \ref{thm-hsg}). $\overline{\rm VP}_{\ol{\F}_q}$ consists of $n$-variate degree $d=n^{O(1)}$ polynomials, over the field $\ol{\F}_q$, that can be `infinitesimally approximated' by size $s=n^{O(1)}$ algebraic circuits. This problem is interesting as natural questions like explicit construction of the normalization map (in Noether's Normalization Lemma NNL) reduce to the construction of a hitting-set of $\overline{\rm VP}$ \cite{mulmuley2017geometric}; which was previously known to be only in EXPSPACE \cite{mulmuley2017geometric,mul12}. 
  This was recently improved greatly, over the field $\mathbb{C}$, by  \cite{forbes2017pspace}. Their proof technique uses real analysis and does not apply to finite fields. We need to develop purely algebraic concepts to solve the finite field case (namely through AnnAtZero), which then apply to {\em any} field. 

\smallskip 
To further motivate the concept of algebraic dependence, we list a few recent problems in computer science.
 The first problem is about constructing an explicit randomness extractor for sources which are polynomial maps over finite fields. Using Jacobian criterion, \cite{DGW09, bib:Dvi09} solved the problem for fields with large characteristic.
 The second application is in the famous polynomial identity testing (PIT) problem. To efficiently design hitting-sets, for some interesting models, 
 \cite{BMS13, ASSS11, KS15} constructed a family of trdeg-preserving maps. For more background and applications of algebraic dependence testing, see \cite{pandey2016algebraic}. The annihilator has been a key concept to prove the connection between hitting-sets and lower bounds \cite{heintz1980testing}, and in  bootstrapping `weak' hitting-sets \cite{AGS17}.

\vspace{-0.05in}
\subsection{Our results}\label{contrib}
\vspace{-0.05in}

 In this paper, we give Arthur-Merlin protocols \& algorithms, with proofs using basic tools from algebraic geometry. The first theorem we prove is about AD($\F_q$).

\begin{theorem}\label{thm_amcoam}
Algebraic dependence testing of circuits in $\F_q[\mathbf{x}]$ is in AM $\cap$ coAM.
\end{theorem}
 
This result vastly improves the current best upper bound known for AD($\F_q$)-- from being `outside' the polynomial hierarchy (namely NP$^{\#\P}$ \cite{mittmann2014algebraic}) to `lower' than the second-level of polynomial hierarchy (namely AM $\cap$ coAM). This rules out the possibility of  AD($\F_q$) being NP-hard (unless polynomial hierarchy collapses to the second-level \cite{AB09}). Recall that, for zero or large characteristic $\F$, AD($\F$) is in coRP (Section \ref{sec-prelim}). We conjecture such a result for AD($\F_q$) too.

\smallskip
Our second result is about the problem AnnAtZero (i.e.~testing whether all the annihilators of given $\mathbf{f}$ have constant term zero). A priori it is unclear why it should  have complexity better than EXPSPACE (note: ideal membership is EXPSPACE-complete \cite{mayr1982complexity}). Firstly, we relate to a (new) version of polynomial system satisfiability, over the algebraic closure $\overline{\F}$:

\begin{problem}[Approximate polynomials satisfiability (APS)]
 Given algebraic circuits $f_1,\ldots,f_m \in \F[x_1,\ldots,x_n]$, does there exist $ \mathbf\beta \in \overline{\F}(\epsilon)^n $ such that for all $i$, $f_i(\mathbf{\beta})$ is in the ideal $\epsilon \overline{\F}[\epsilon] $? If yes, then we say that $\mathbf{f}:= \{f_1,\ldots,f_m\}$ is in APS.
\end{problem}

It is easy to show: Function field $\ol{\F}(\epsilon)$ here can be equivalently replaced by {\em Laurent polynomials} $\overline{\F}[\epsilon, \epsilon^{-1}]$, or, the field $\ol{\F}((\epsilon))$ of {\em formal Laurent series} (use mod $\epsilon \overline{\F}[\epsilon]$). A reason why these objects appear in algebraic complexity can be found in \cite[Sec.5.2]{burgisser2004complexity} \& \cite[Sec.5]{LL89}. They help algebrize the notion of `infinitesimal approximation' (in real analysis think of $\epsilon\rightarrow0$ \& $1/\epsilon\rightarrow\infty$). A notable computational issue involved is that the degree bound of $\epsilon$ required for $\beta$ is exponential in the input size \cite[Prop.3]{LL89}; this may again be a ``justification'' for APS requiring that much space.

\smallskip
Classically, the {\em exact} version of APS has been extremely well-studied-- Does there exist $\mathbf{\beta} \in \overline{\F}^n$ such that for all $i$, $f_i(\mathbf{\beta})=0 $? This is what Hilbert's Nullstellensatz (HN) characterizes and yields an impressive PSPACE algorithm \cite{koiran1996hilbert, kollar1988sharp}. Note that if system $\mathbf{f}$ has an exact solution, then it is trivially in APS. But the converse is not true. For example,  $\{x, xy-1\}$ is in APS, but there is no exact solution in $\overline{\F}$. To see the former, assign $x=\epsilon$ and $y= 1/\epsilon$. Also, the instance $\{x,x+1\}$ is neither in APS nor has an exact solution. Finally, note that if we restrict $ \mathbf{\beta}$ to come from $\overline{\F}[\epsilon]^n$ then APS becomes equivalent to exact satisfiability and HN applies. This can be seen by going modulo $\epsilon \overline{\F}[\epsilon]$, as the quotient $\overline{\F}[\epsilon]/\epsilon \overline{\F}[\epsilon]$ is $\overline{\F}$.

Coming back to AnnAtZero, we show that it is equivalent both to a geometric question and to deciding APS. This gives us, with more work, the following surprising consequence.

\begin{theorem}\label{thm-aps}
 APS is NP-hard and is in PSPACE.
\end{theorem}

We apply this to design hitting-sets  and solving NNL (refer \cite{mulmuley2017geometric} for the background).

\begin{theorem}\label{thm-hsg}
There is a PSPACE algorithm that (given input $n,s,r$ in unary \& suitably large $\F_q$) outputs a set, of points from $\F_q ^n$ of size poly$(nsr, \log q)$, that hits all $n$-variate degree-$r$ polynomials over $\ol{\F}_q $ that can be infinitesimally approximated by size $s$ circuits.
\end{theorem}

\noindent
{\bf More applications?} The exact polynomials satisfiability question HN (over $\ol{\F}$) is highly expressive and, naturally, most computer science problems get expressed that way. We claim that in a similar spirit, the APS question expresses those computer science problems that involve `infinitesimal approximation'. 
One prominent example is the concept of {\em border rank} of tensor polynomials (used in matrix multiplication algorithms and GCT, see \cite{burgisser2013algebraic, landsberg2012tensors, le2014powers}). Border rank computation of a given tensor (over $\ol{\F}$) can easily be reduced to an APS instance and, hence, now solved in PSPACE; this matches the complexity of tensor rank itself \cite{schaefer2016complexity}. From the point of view of Gr\"obner basis theory, APS is a problem that seems a priori much harder than HN. Now that both of them have a PSPACE algorithm, one may wonder whether it can be brought all the way down to NP or AM? (In fact, $\text{HN}_{\C}$ is known to be in AM, conditionally under GRH \cite{koiran1996hilbert}.)

Our methods in the proof of Theorem \ref{thm-aps} imply an interesting ``degree bound'' related to the (prime) ideal $I$ of annihilators of polynomials $\mathbf f$. Namely, $I=\sqrt{I_{\le d}}$, where $I_{\le d}$ refers to the subideal generated by degree $\le d$ polynomials of $I$, $d$ is the Perron-like bound $(\max_{i\in [m]}\deg(f_i))^k$, and $k:= \text{ trdeg}(\mathbf f)$. This is equivalent to the geometric fact, which we prove, that the varieties defined by the two ideals $I$ and $I_{\le d}$ are equal (Theorem \ref{thm-randReduct}). This again is an exponential improvement over what one expects to get from the general Gr\"obner basis methods; because, the generators of $I$ may well have doubly-exponential degree.

The hitting-set result (Theorem \ref{thm-hsg}) can be applied to compute, in PSPACE, the explicit system of parameters (esop) of the {\em invariant ring} of the variety $\Delta[\det,s]$, over $\ol{\F}_q$, with a given group action \cite[Thm.4.9]{mulmuley2017geometric}. Also, we can now construct, in PSPACE, polynomials in $\F_q[x_1,\dots,x_n]$ that cannot even be approximated by `small' algebraic circuits. Such results were previously known only for characteristic zero fields, see \cite[Thms.1.1-1.4]{forbes2017pspace}. Bringing this complexity down to P is the longstanding problem of blackbox PIT (\& lower bounds), see \cite{Saxena09, shpilka2010arithmetic, Saxena13}. Mulmuley \cite{mul12} pointed out that small hitting-sets for $\ol{\rm VP}$ can be designed in EXPSPACE which is a far worse complexity than that for VP. He called it the GCT Chasm. We bridge it somewhat, as the proof of Theorem \ref{thm-hsg} shows that small hitting-sets for $\ol{\rm VP}_{\ol{\F}}$ can be designed in PSPACE (like those for VP) for {\em any} field $\F$.


\vspace{-0.05in}
 \subsection{Proof ideas}\label{idea}
\vspace{-0.05in}
 
\smallskip\noindent
\textbf{Proof idea of Theorem \ref{thm_amcoam}.}  
Suppose we are given algebraic circuits $\mathbf{f}:= \{f_1,\ldots,f_m\}$ computing in $\F_q[x_1,\dots,x_n]$. For the AM and coAM protocols, we consider the following system of equations over a `small' extension $\F_{q'}$:
  
 For $b=(b_1,\dots,b_n)\in \F_{q'}^n$, define the system of equations $f_i(x_1,\dots,x_n)=b_i$, for $i\in [m]$. We denote the number of solutions of the above system in $\F_{q'}^n$ as $N_b$. Let $f:\F_{q'}^n\to \F_{q'}^m$ be the polynomial map $a\mapsto (f_1(a),\dots,f_m(a))$.

{\em AM gap}. [Theorem \ref{thm_am}] We establish bounds for the number $N_{f(a)}$, where $a$ is a random point in $\F_{q'}^n$. If $f_1,\dots,f_m$ are independent, we show that  $N_{f(a)}$ is relatively small. Whereas, if the polynomials are algebraically dependent then $N_{f(a)}$ is much more. 

Assume $\mathbf{f}$ are algebraically independent. Wlog (see the full version of    
\cite[Sec.2]{pandey2016algebraic}) we can assume that $m=n$ and for all $i\in[n]$, $\{x_i,f_1,\ldots,f_n\}$ are algebraically dependent. 
The first step is to show that the zeroset defined by the system of equations, for random $f(a)$, has dimension $\leq 0$. This is proved using the Perron degree bound on the annihilator of $\{x_i,f_1,\ldots,f_n\}$. 
Next, one can apply an affine version of Bezout's theorem to upper bound $N_{f(a)}$. 
On the other hand, suppose $\mathbf{f}$ are algebraically dependent, say with annihilator $Q$. Let $\im(f):=f(\F_{q'}^n)$ be the image of $f$. Since $Q$ vanishes on $\im(f)$, we know that $\im(f)$ is relatively small, whence we deduce that $N_{f(a)}$ is large for `most' $a$'s.
 
{\em coAM gap}. [Theorem \ref{thm_coam}] We pick a random point $b=(b_1,\dots,b_m)\in \F_{q'}^m$ and bound $N_b$, which is the number of solutions of the system defined above. In the dependent case, we show that $N_b=0$ for `most' $b$'s. But in the independent case, we show that $N_b\ge1$ for `many' (may be not `most'!) $b$'s. The ideas are based on those sketched above. 

The two kinds of gaps shown above are based on the set 
$f^{-1}(f(\mathbf x))$ resp.~$\im(f)$. Note that membership in either of these sets is testable in NP (the latter requires nondeterminism). Based on this and the gaps between the respective cardinalities, we can invoke Lemma \ref{lem_amprotocol} and devise the AM and coAM protocols for AD($\F_{q'}$), which also apply to AD($\F_{q}$).
 
{\em Remark--} One advantage in our problem is that we could sample a random point in the set $\im(f)$. In contrast, it is not clear how to sample a random point in the zeroset $\zer(\mathbf f):= \{\mathbf x\in \F_{q'}^n \,\mid\, f(\mathbf x)= \mathbf 0\}$. Thus, we manage to side-step the NP-hardness associated with most zeroset properties. Eg.~computing the dimension of $\zer(\mathbf f)$ is NP-hard.

\smallskip\noindent
\textbf{Proof idea of Theorem \ref {thm-aps}.} 
Let algebraic circuits $\mathbf f:= \{f_1,\dots,f_m\}$ in $\F[x_1,\dots,x_n]$ be given over a field $\F$.
We want to determine if the constant term of every annihilator for $\mathbf f$ is zero. Redefine the polynomial map $f:\ol{\F}^n\to \ol{\F}^m$; $a\mapsto (f_1(a),\dots,f_m(a))$. For a subset $S$ of an affine (resp.~projective) space, write $\overline{S}$ for its {\em Zariski closure} in that space, i.e.~it is the smallest subset that contains $S$ and equals the zeroset $\zer(I)$ of some polynomial ideal $I$. 

{\em APS vs AnnAtZero.} [Theorem \ref{thm_approx}] Now, we interpret the problem AnnAtZero in a geometric way through Lemma \ref{lem_geom}:

 The constant term of every annihilator of $\mathbf f$ is zero iff the origin point $\mathbf 0\in \ol{\im(f)}$. 
 
This has a simple proof using the ideal-variety correspondence \cite{Har92}. Note that the stronger condition $\mathbf 0\in\im(f)$ is equivalent to the existence of a common solution to the equations $f_i(x_1,\dots,x_n)$ $=0$, $i=1,\dots,m$. The latter problem (call it HN for Hilbert's Nullstellensatz) is known to be in AM if  $\F=\mathbb{Q}$ and GRH is assumed \cite{koiran1996hilbert}.
However, $\im(f)$ is not necessarily Zariski closed; equivalently, it may be strictly smaller than $\ol{\im(f)}$. So, we need new ideas to test $\mathbf{0}\in \overline{\im(f)}$. 
  
Next, we observe that although $\mathbf{0}\in \overline{\im(f)}$ is not equivalent to the existence of a solution $\mathbf x\in \ol{\F}^n$ to $f(\mathbf x)= \mathbf 0$, it {\em is} equivalent to the existence of an  ``approximate solution'' $\mathbf x\in \ol{\F}(\epsilon)^n$, which is an $n$-tuple of rational functions in a formal variable $\epsilon$. The proof idea of this uses a degree bound on $\epsilon$ due to \cite{LL89}. We called this problem APS. As AnnAtZero problem is already known to be NP-hard \cite{Kay09}, APS is also NP-hard. 

{\em Upper bounding APS.} We now know that: Solving APS for $\mathbf f$ is equivalent to solving AnnAtZero for $\mathbf f$. AnnAtZero was previously known to be in PSPACE in the special case when the trdeg $k$ of $\F(\mathbf f)/\F$ equals $m$ or $m-1$, but the general case remained open (best being EXPSPACE).

In this work we prove that AnnAtZero is in PSPACE even when $k<m-1$. Our simple idea is to reduce the input to a smaller $m=k+1$ instance, by choosing new polynomials $g_1,\dots,g_{k+1}$ that are random linear combinations of $f_i$'s. We show that with high probability, replacing $\{f_1,\dots,f_m\}$ by $\{g_1,\dots,g_{k+1}\}$ preserves YES/NO instances as well as the trdeg. This gives a randomized poly-time reduction from the case $k<m-1$ to $k=m-1$ (Theorem \ref{thm-randReduct}). The latter has a standard PSPACE algorithm.

For notational convenience view $\ol{\F}$ as the {\em affine line} $\aff$.
Define $V:=\overline{\im(f)}\subseteq \aff^m$. Proving that the above reduction (of $m$) does preserve YES/NO instances amounts to proving the following geometric statement: If $V$ does not contain the origin $O\in\aff^m$, then with high probability, the variety $V':=\overline{\pi(V)}$ does not contain the origin $O'\in\aff^{k+1}$ either, where $\pi:\aff^m\to \aff^{k+1}$ is a random linear map. 

As $\pi$ is picked at random, the kernel $W$ of $\pi$ is a random linear subspace of $\aff^m$.
We have $O'\not\in \pi(V)$ whenever $V\cap  W=\emptyset$, but this is not sufficient for proving $O'\not\in \overline{\pi(V)}$, since $V$ may ``get arbitrarily close to $W$'' in $\aff^m$ and meet $W$ ``at infinity''. Inspired by this observation, we consider projective geometry instead of affine geometry, and prove that $O'\not\in V'$ holds as long as the projective closure of $V$ and that of $W$ are disjoint. The proof uses the construction of a projective subvariety-- the {\em join} --to characterize $\pi^{-1}(V')$, and eventually rules out $W\subseteq \pi^{-1}(V')$ (Lemma \ref{lem_projcriterion}).

Moreover, we show that this holds with high probability if $O\not\in V$: by (repeatedly) using the fact that a generic (=random) hyperplane section reduces the dimension of a variety by one.

\smallskip\noindent
\textbf{Proof idea of Theorem \ref{thm-hsg}.} 
Define $\aff:=\ol{\F}_q$ and assume wlog $q\ge\Omega(sr^{2})$ \cite{AL86}.
\cite[Thm.4.4]{heintz1980testing} showed that a hitting-set, of size $h:= O(s^2n^2\log q)$ in $\F_q^n$, {\em exists} for the class of degree-$r$ polynomials, in $\aff[x_1,\dots,x_n]$, that can be infinitesimally approximated by size-$s$ algebraic circuits. So, we can search over all possible subsets of size $h$ from $\F_q^n$ and `most' of them are hitting-sets.

How do we certify that a candidate set $\calH$ is a hitting-set? The idea is to use universal circuits. A {\em universal circuit} has $n$ essential variables $\mathbf x=\{x_1,\ldots,x_n\}$ and $s':=O(sr^4)$ auxiliary variables $\mathbf y= \{y_1,\ldots, y_{s'}\}$. We can fix the auxiliary variables, from $\aff(\epsilon)$, in such a way so that it can output any homogeneous circuit of size-$s$, approximating a degree-$r$ polynomial in $\ol{\rm VP}_\aff$. Given a universal circuit $\Psi$, certification of a hitting-set $\calH$ is based on the following observation, that follows from the definitions: 

Candidate set $\calH=: \{\mathbf v_1,\ldots, \mathbf v_h\}$ is a hitting-set iff~~  
$\forall \mathbf y\in \aff(\epsilon)^{s'} $, $\,\Psi(\mathbf{y,x})\notin \epsilon\aff[\epsilon][\mathbf x] $ $ \,\Rightarrow\, $ $\exists i\in[h]$, $ \Psi(\mathbf y, \mathbf v_i)\notin \epsilon\aff[\epsilon] $.

\smallskip
Equivalently: Candidate set $\calH= \{\mathbf v_1,\ldots, \mathbf v_h\}$ is {\em not} a hitting-set iff~~  
$\exists \mathbf y\in \aff(\epsilon)^{s'} $, $\,\Psi(\mathbf{y,x})\notin \epsilon\aff[\epsilon][\mathbf x]\,$ and  $\,\forall i\in[h]$, $ \Psi(\mathbf y, \mathbf v_i)\in \epsilon\aff[\epsilon] $.

\smallskip
Note that this hitting-set certification is more challenging than the one against polynomials in VP; because the degree bounds for $\epsilon$ are exponentially high and moreover, we do not know how to frame the first `non-containment' condition as an APS instance. To translate it to an APS instance, our key idea is the following. 

Pick $q\ge\Omega(s'r^2)$ so that a hitting-set exists, in $\F_q^n$, that works against polynomials approximated by the specializations of $\Psi$. 
Suppose $\Psi(\mathbf{\alpha,x})$ is not in $\epsilon\aff[\epsilon][\mathbf x]$, for some $\mathbf \alpha\in \aff(\epsilon)^{s'}$. This means that we can write it as $\sum_{-m\le i\le m'}$ $\epsilon^i g_i(\mathbf x)$ with $g_{-m}\ne0$ and $m\ge0$. Clearly, $\epsilon^m\cdot \Psi(\mathbf{\alpha,x})$ infinitesimally approximates the nonzero polynomial $g_{-m} \in \aff[\mathbf x]$.
By the conditions on $\Psi$, we know that $g_{-m}$ is a homogeneous degree-$r$ polynomial (and approximative complexity $s'$). Thus, by \cite{Sch80}, there exists a $\mathbf \beta\in\F_q^n$ such that $g_{-m}(\mathbf \beta)=:a$ is a nonzero element in $\aff$. We can normalize by this and consider $a^{-1}\epsilon^m\cdot \Psi(\mathbf{y,x})$, which evaluates to $1+\epsilon\aff[\epsilon]$ at $(\mathbf{\alpha,\beta})$. Since this normalization factor only affects the auxiliary variables $\mathbf y$, we get another equivalent criterion:

Candidate set $\calH= \{\mathbf v_1,\ldots, \mathbf v_h\}$ is {\em not} a hitting-set iff~~  
$\exists \mathbf y\in \aff(\epsilon)^{s'} $ and $\exists \mathbf x\in \F_q^n $ such that, $\,\Psi(\mathbf{y,x}) - 1 \in \epsilon\aff[\epsilon]\,$ and  $\,\forall i\in[h]$, $ \Psi(\mathbf y, \mathbf v_i)\in \epsilon\aff[\epsilon] $.

\smallskip
We reached closer to APS, but how do we implement $\exists? \mathbf x\in \F_q^n $ (it is an exponential space)? 

The idea is to rewrite it, instead using the $(r+1)$-th roots of unity $Z_{r+1}\subset\aff$, as: $\exists \mathbf x\in \aff(\epsilon)^n$, $\forall i\in[n]$, $x_i^{r+1}-1 \in \epsilon\aff[\epsilon]$. This gives us a criterion that is an instance of APS with $n+h+1$ input polynomials (Theorem \ref{thm_HITTING}). By Theorem \ref{thm-aps} it can be done in PSPACE; finishing the proof. Moreover, this PSPACE algorithm idea is independent of the field characteristic. (Eg.~it can be seen as an alternative to \cite{forbes2017pspace} over the complex field.)

\vspace{-0.05in}
\section{Preliminaries}\label{sec-prelim}
\vspace{-0.05in}

{\bf Jacobian.}
Although this work would not need it, we define the classical Jacobian: For polynomials $\mathbf{f} = \inbrace{f_1, \cdots, f_m}$ in $\F[x_1,\cdots, x_n]$, {\em Jacobian} is the matrix $\mathcal{J}_{\mathbf{x}}(\mathbf{f}) := (\pder{f_i}{x_j})_{m \times n}$, where $\pder{f_i}{x_j} := \partial f_i/ \partial x_j$. 

Jacobian criterion \cite{Jac41,BMS13} states: For degree $\le d$ and trdeg $\le r$ polynomials $\mathbf{f}$, if char$(\F) = 0$ or char$(\F) > d^r$, 
then  trdeg$(\mathbf{f}) \,=\, \rk_{\F(\mathbf{x})} \mathcal{J}_{\mathbf{x}}(\mathbf{f})$. This yields a randomized poly-time algorithm \cite{Sch80}. For other fields, Jacobian criterion fails due to inseparability and AD($\F$) is open.

\smallskip\noindent
\textbf{AM protocol.}
Arthur-Merlin class AM is a randomized version of the class NP (see \cite{AB09}).
Arthur-Merlin protocols, introduced by Babai \cite{babai1985trading}, can be considered as a special type of interactive proof system in which the randomized poly-time verifier (Arthur) and the all-powerful prover (Merlin) have only constantly many rounds of exchange. 
AM contains interesting problems like determining if two graphs are non-isomorphic.
AM $\cap$ coAM is the class of decision problems for which both YES and NO answers can be verified by an AM protocol. It can be thought of as the randomized version of NP $\cap$ coNP. See \cite{kayal2006complexity} for a few natural algebraic problems in AM $\cap$ coAM. If such a problem is NP-hard (even under random reductions) then polynomial hierarchy collapses to the second-level, i.e.~PH$=\Sigma_2$.

In this work AM protocol will only be used to distinguish whether a set $S$ is `small' or `large'. Formally, we refer to the Goldwasser-Sipser Set Lowerbound method:
\begin{lemma}\cite[Chap.9]{AB09} \label{lem_amprotocol}
Let $m\in\N$ be given in binary. Suppose $S$ is a set whose membership can be tested in nondeterministic polynomial time and its size is promised to be either $\le m$ or $\ge 2m$. Then, the problem of deciding whether $|S|\stackrel{?}{\ge}2m$ is in AM.
\end{lemma}

\smallskip\noindent
\textbf{Geometry.} Due to limited space we have moved the geometry preliminaries to Appendix \ref{app-AG}. One can also refer to a standard text, eg.~\cite{Har92, hartshorne2013algebraic}. Basically, we need terms about affine (resp.~projective) zerosets and the underlying Zariski topology. The latter gives a way to `impose' geometry even in very discrete situations, eg.~finite fields in this work. 

\vspace{-0.05in}
\section{Algebraic dependence testing: Proof of Theorem \ref{thm_amcoam} }\label{sec-dep}
\vspace{-0.03in}

Given $f_1,\dots,f_m\in \F_q[x_1,\dots,x_n]$, we want to decide if they are algebraically dependent. For this problem AD($\F_q$) we could assume, with some preprocessing, that $m=n$. For, $m>n$ means that its a YES instance. If $m<n$ then we could apply a `random' linear map on the variables to reduce them to $m$, preserving the YES/NO instances. Also, the trdeg does not change when we move to the algebraic closure $\ol{\F}_q$. The details can be found in \cite[Lem.2.7-2.9]{pandey2016algebraic}. So, we assume the input instance to be $\mathbf f := \inbrace{f_1,\dots,f_n}$ with nonconstant polynomials.
 
 In the following, let $D:=\prod_{i\in [n]}\deg(f_i) \,>0$ and $D':=\max_{i\in [n]} \deg(f_i) \,>0$. Let $d\in\N^+$ and $q'=q^d$. The value of $d$ will be determined later.
Let $f:\F_{q'}^n\to\F_{q'}^n$ be the polynomial map $a\mapsto (f_1(a),\dots,f_n(a))$. For $b=(b_1,\dots,b_n)\in \F_{q'}^n$, denote by $N_b$ the size of the preimage $f^{-1}(b)=$ $\{\mathbf x \in \F_{q'}^n \,\mid\, f(\mathbf x) = b \}$.

Define $\aff:=\ol{\F}_q$ and $N'_b:= \#\{\mathbf x \in \aff^n \,\mid\, f_i(\mathbf x) = b_i, \, \text{ for all } i\in[n] \}$ which might be $\infty$.
Let $Q\in\F_q[y_1,\dots,y_n]$ be a nonzero annihilator, of minimal degree, of $f_1,\dots,f_n$. If it exists then $\deg(Q)\leq D$ by Perron's bound. 

\subsection{AM protocol}

First, we study the independent case.

\begin{lemma}[Dim=0 preimage]\label{lem_finite}
Suppose $\mathbf f$ are independent. Then $N'_{f(a)}$ is {\em finite} for all but at most $(nDD'/q')$-fraction of $a\in\F_{q'}^n$.
\end{lemma}
\begin{proof}
For $i\in [n]$, let $G_i\in\F_q[z, y_1,\dots,y_n]$ be the annihilator of $\inbrace{x_i, f_1,\dots,f_n}$. We have $\deg(G_i)\leq D$ by Perron's bound. 
Consider $a\in\F_{q'}^n$ such that  $G'_i(z):= G_i(z, f_1(a),\dots,f_n(a))\in\F_q[z]$ is a nonzero polynomial for every $i\in [n]$.
We claim that $N'_{f(a)}$ is finite for such $a$. 

To see this, note that for any $b=(b_1,\dots,b_n)\in \aff^n$ satisfying the equations $f_i(b)=f_i(a)$, $i\in[n]$, we have
\[
0=\, G_i(b_i, f_1(b), \dots, f_n(b)) \,=\, G_i(b_i, f_1(a),\dots, f_n(a)) \,=\, G'_i(b_i), \quad\forall i\in [n] \,.
\]
Hence, each $b_i$ is a root of $G'_i$. It follows that $N'_{f(a)} \leq \prod_{i\in[n]} \deg(G'_i) < \infty$, as claimed. 

It remains to prove  that the number of $a\in\F_{q'}^n$ satisfying $G'_i=0$, for some index $i\in [n]$, is bounded by $nDD'q'^{-1}\cdot q'^{n}$.
Fix $i\in [n]$. Suppose $G_i=\sum_{j=0}^{d_i} G_{i,j} z^j$, where $d_i:=\deg_z(G_i)$ and $G_{i,j}\in \F_q[y_1,\dots,y_n]$, for $0\leq j\leq d_i$.
The leading coefficient $G_{i,d_i}$ is nonzero. As $f_1,\dots,f_n$ are algebraically independent, the polynomial $G_{i,d_i}(f_1,\dots,f_n)\in\F_q[x_1,\dots,x_n]$ is also nonzero.
Its degree is $\le D'\deg(G_{i,d_i}) \leq D'\deg(G_i) \leq DD'$.
By \cite{Sch80}, for all but at most $(DD'/q')$-fraction of $a\in \F_{q'}^n$, we have $G_{i,d_i}(f_1(a),\dots,f_n(a))\neq 0$ which   implies  
\[
G'_i(z) \,=\, G_i(z, f_1(a),\dots,f_n(a)) \,=\, \sum_{j=0}^{d_i} G_{i,j}(f_1(a),\dots,f_n(a)) z^j \,\neq 0 \,.
\]
The claim now follows from the union bound.
\end{proof}

We need the following affine version of B\'ezout's Theorem. Its proof can be found in \cite[Thm.3.1]{Sch95}.
\begin{theorem}[B\'ezout's]
 Let $g_1,\dots,g_n\in \aff[x_1,\dots,x_n]$. Then the number of common zeros of $g_1,\dots,g_n$ in $\aff^n$ is either infinite, or at most $\prod_{i\in[n]} \deg(g_i)$.
\end{theorem}

Combining Lemma~\ref{lem_finite} with B\'ezout's Theorem, we obtain
\begin{lemma}[Small preimage]\label{lem_bpre}
Suppose $\mathbf f$ are independent.
Then $N_{f(a)}\leq D$ for all but at most $(nDD'/q')$-fraction of $a\in\F_{q'}^n$.
\end{lemma}

Next, we study the dependent case (with an annihilator $Q$).

\begin{lemma}[Large preimage]\label{lem_coamdep}
Suppose $\mathbf f$ are dependent.
Then for $k>0$, we have $N_{f(a)}>k$ for all but at most $(kD/q')$-fraction of $a\in \F_{q'}^n$.
\end{lemma}
\begin{proof}
 Let $\im(f):= f(\F_{q'}^n)$ be the image of the map. Note that $Q$ vanishes on all the points in $\im(f)$. So, $|\im(f)| \leq Dq'^{n-1}$ by \cite{Sch80}.
 
Let $B:= \{b\in \im(f): N_b\leq k\}$ be the ``bad'' images. We can estimate the bad domain points as, 
\[
\#\{a\in \F_{q'}^n: N_{f(a)}\leq k\} \,=\, \#\{a\in \F_{q'}^n: f(a)\in B\} \,\le\,  k|B| \le k|\im(f)| \,\le\, k D q'^{n-1} \,.
\]
which proves the lemma.
\end{proof}

\begin{theorem}[AM]\label{thm_am}
Testing algebraic dependence of $\mathbf f$ is in AM.
\end{theorem}
\begin{proof}
Fix $q'=q^d>4nDD'+4kD$ and $k:=2D$. Note that $d$ will be polynomial in the input size.
For an $a\in \F_{q'}^n$, consider the set 
$f^{-1}(f(a)) :=$ $\{\mathbf x \in \F_{q'}^n \,\mid\, f(\mathbf x) = f(a) \}$.

By Lemmas \ref{lem_bpre} \& \ref{lem_coamdep}: When Arthur picks $a$ randomly, with high probability, $|f^{-1}(f(a))|= N_{f(a)}$ is more than $2D$ in the dependent case while $\le D$ in the independent case. Note that an upper bound on $\prod_{i\in [n]}\deg(f_i)$ can be deduced from the size of the input circuits for $f_i$'s; thus, we know $D$. Moreover, containment in $f^{-1}(f(a))$ can be tested in P. Thus, by Lemma \ref{lem_amprotocol}, AD($\F_q$) is in AM.
\end{proof}

\subsection{coAM protocol}

We again study the independent case wrt a different point in the range of $f$.

\begin{lemma}[Large image]\label{lem-coam-1}
Suppose $\mathbf f$ are independent.
Then $N_b>0$ for at least $(D^{-1} - nD'q'^{-1})$-fraction of $b\in \F_{q'}^n$.
\end{lemma}
\begin{proof}
Let $S:= \{a\in \F_{q'}^n: N_{f(a)}\leq D\}$.
Then $|S|\geq (1-n DD'q'^{-1})\cdot q'^{n}$ by Lemma \ref{lem_bpre}.
As every $b\in f(S)$ has at most $D$ preimages in $S$ under $f$, we have 
$|f(S)| \,\ge\, |S|/D \,\ge\, (D^{-1}-nD'q'^{-1})\cdot q'^{n} $.
This proves the lemma since $N_b>0$ for all $b\in f(S)$.
\end{proof}

Next, we study the dependent case.

\begin{lemma}[Small image]\label{lem-coam-2}
Suppose $\mathbf f$ are dependent.
Then $N_b=0$ for all but at most $(D/q')$-fraction of $b\in \F_{q'}^n$.
\end{lemma}
\begin{proof}
By definition: $N_b>0$ iff $b\in \im(f):=f(\F_{q'}^n)$. It was shown in the proof of Lemma~\ref{lem_coamdep} that $|\im(f)| \le Dq'^{n-1}$. The lemma follows.
\end{proof}

\begin{theorem}[coAM]\label{thm_coam}
Testing algebraic dependence of $\mathbf f$ is in coAM.
\end{theorem}
\begin{proof}
Fix $q'=q^d \,> D(2D+nD')$. Note that $d$ will be polynomial in the input size.
For $b\in \F_{q'}^n$, consider the set 
$f^{-1}(b) :=$ $\{\mathbf x \in \F_{q'}^n \,\mid\, f(\mathbf x) = b \}$ of size $N_b$.

Define $S:= \im(f)$. Note that: $b\in \F_{q'}^n$ has $N_b>0$ iff $b\in S$. Thus, by Lemma \ref{lem-coam-1} (resp.~Lemma \ref{lem-coam-2}), $|S|\ge (D^{-1} - nD'q'^{-1})q'^n\,> 2Dq'^{n-1}$ (resp.~$|S|\le Dq'^{n-1}$) when $\mathbf f$ are independent (resp.~dependent). 
Note that an upper bound on $\prod_{i\in [n]}\deg(f_i)$ can be deduced from the size of the input circuits for $f_i$'s; thus, we know $Dq'^{n-1}$. Moreover, containment in $S$ can be tested in NP. Thus, by Lemma \ref{lem_amprotocol}, AD($\F_q$) is in coAM.
\end{proof}

\begin{proof}[Proof of Theorem \ref{thm_amcoam}]
The statement immediately follows from Theorems \ref{thm_am} \& \ref{thm_coam}.
\end{proof}

\vspace{-0.05in}
\section{Approximate polynomials satisfiability: Proof of Theorem \ref{thm-aps}} \label{sec:aps}
\vspace{-0.05in}

Theorem \ref{thm-aps} is proved in two parts. First, we show that APS is equivalent to AnnAtZero problem; which means that it is NP-hard \cite{Kay09}. Next, we utilize the  beautiful underlying geometry to devise a PSPACE algorithm.

\subsection{APS is equivalent to AnnAtZero}
\vspace{-0.05in}

Let $\aff$ be the algebraic closure of $\F$. Note that for the given polynomials $\mathbf f:= \inbrace{f_1,\dots,f_m}$ in $\F[\mathbf x]$, there is an annihilator over $\F$ with nonzero constant term iff there is an annihilator over $\aff$ with nonzero constant term. This is because if $Q$ is an annihilator over $\aff$ with nonzero constant term, wlog $1$, then (by basic linear algebra) the linear system in terms of the (unknown) coefficients of $Q$ would also have a solution in $\F$. Thus, there is an annihilator over $\F$ with constant term $1$. This proves that it suffices to solve AnnAtZero over the algebraically closed field $\aff$. This provides us with a better geometry. 

 Write $f:\aff^n\to \aff^m$ for the polynomial map sending a point $x=(x_1,\dots,x_n)\in \aff^n$ to $(f_1(x),\dots,f_m(x))\in \aff^m$.
 For a subset $S$ of an affine or projective space, write $\overline{S}$ for its Zariski closure in that space. We will use $O$ to denote the origin $\mathbf 0$ of an affine space.

 
The following lemma reinterprets APS in a geometric way.
 \begin{lemma}[$O$ in the closure]\label{lem_geom}
 The constant term of every annihilator for $\mathbf f$ is zero ~iff ~$O\in \overline{\im(f)}$.
 \end{lemma}
 \begin{proof}
 Note that: $Q\in \aff[Y_1,\dots,Y_m]$ vanishes on $\im(f)$ iff $Q(\mathbf f)$ vanishes on $\aff^n$, which holds iff $Q(\mathbf f)=0$, i.e., $Q$ is an annihilator for $\mathbf f$. So $\overline{\im(f)}=V(I)$, where the ideal $I\subseteq \aff[Y_1,\dots,Y_m]$ consists of the annihilators for $\mathbf f$. Also note that $\{O\}=V(\m)$, where $\m$ is the maximal ideal $\langle Y_1,\dots, Y_m\rangle$.

Let us study the condition $O\in \overline{\im(f)}$.
 By the ideal-variety correspondence, $\{O\} \,=\, V(\m) \,\subseteq\, \overline{\im(f)} \,=\, V(I)$ is equivalent to $I\subseteq \m$, i.e.,  $Q\bmod \m=0$ for  $Q\in I$. But $Q\bmod \m$ is just the constant term of the annihilator $Q$. Hence, we have the equivalence.
 \end{proof}
 
As an interesting corner case, the above lemma proves that whenever $\mathbf f$ are algebraically {\em independent}, we have $\aff^m=\ol{\im(f)}$. Eg.~$f_1=X_1$ and $f_2=X_1 X_2-1$. Even in the dependent cases, $\im(f)$ is not necessarily closed in the Zariski topology. 
\begin{exmp}
Let $n=2$, $m=3$. Consider $f_1=f_2=X_1$ and $f_3=X_1 X_2-1$. The annihilators are multiples of $(Y_1-Y_2)$, which means by Lemma \ref{lem_geom} that $O\in \ol{\im(f)}$. But there is no solution to $f_1=f_2=f_3=0$, i.e.~$O\notin \im(f)$.
\end{exmp}

\noindent
\textbf{Approximation.} Although $O\in \overline{\im(f)}$ is not equivalent to the existence of a solution $x\in \aff^n$ to $f_i=0$, $i\in[m]$, it is equivalent to the existence of an  ``approximate solution'' $x\in \aff[\epsilon,\epsilon^{-1}]^n$, which is a tuple of Laurent polynomials in a formal variable $\epsilon$. The formal statement is as follows. Wlog we assume $\mathbf f$ to be $m$ nonconstant polynomials.

\begin{theorem}[Approx.~wrt $\epsilon$]\label{thm_approx}
$O\in \overline{\im(f)}$ iff there exists $x=(x_1,\dots,x_n)\in \aff(\epsilon)^n$ such that $f_i(x)\in \epsilon \aff[\epsilon]$, for all $i\in [m]$. Moreover, when such $x$ exists, it may be chosen such that
\[
x_i \,\in\,\, \epsilon^{-D} \aff[\epsilon] \,\cap\, \epsilon^{D'} \aff[\epsilon^{-1}] \,=\, \left\{\sum_{j=-D}^{D'} c_j \epsilon^j: c_j\in \aff\right\}, \quad i\in[n],
\]
where $D:=\prod_{i\in[m]} \deg(f_i) \,>0$ and $D':=(\max_{i\in[m]} \deg(f_i))\cdot D \,>0$.
\end{theorem}

The proof of Theorem~\ref{thm_approx} is almost the same as that in \cite{LL89}. First, we recall a tool to reduce the domain from a variety to a curve, proven in \cite{LL89}.

\begin{lemma}{\em \cite[Prop.1]{LL89}}\label{lem_aux1}
Let $V\subseteq \aff^n$, $W\subseteq \aff^m$ be affine varieties, $\varphi: V\to W$ dominant, and $t\in W\setminus \varphi(V)$.
Then there exists a curve $C\subseteq \aff^n$ such that $t\in\overline{\varphi(C)}$ and $\deg(C)\leq \deg(\Gamma_\varphi)$, where $\Gamma_\varphi$ denotes the graph of $\varphi$ embedded in $\aff^n\times \aff^m$.
\end{lemma}

Next, \cite{LL89} essentially shows that in the case of a curve one can approximate the preimage of $f$ by using a {\em single} formal variable $\epsilon$ and working in $\aff(\epsilon)$.

\begin{lemma}{\em \cite[Cor.~of Prop.3]{LL89}}\label{lem_aux2}
Let $C\subseteq \aff^n$ be an affine curve.  Let $f: C\to \aff^m$ be a morphism sending $x\in C$ to $(f_1(x),\dots,f_m(x))\in \aff^m$, where $f_1,\dots,f_m\in \aff[X_1,\dots,X_n]$. Let $t=(t_1,\dots,t_m)\in \overline{f(C)}$. Then there exists  $p_1,\dots,p_n\in \epsilon^{-\deg(C)}\aff[[\epsilon]]$ such that $f_i(p_1,\dots,p_n)-t_i \,\in\, \epsilon \aff[[\epsilon]]$ , for all $i\in [m]$.
\end{lemma}

Finally, we can use the above two lemmas to prove the connection of APS with $O\in \overline{\im(f)}$, and hence with AnnAtZero (by Lemma \ref{lem_geom}).

\begin{proof}[Proof of Theorem~\ref{thm_approx}]
First assume that an $x$, satisfying the conditions in Theorem~\ref{thm_approx}, exists. Pick such an $x$. If $\mathbf f$ are algebraically independent then by Lemma \ref{lem_geom} we have that $\aff^m=\overline{\im(f)}$ and we are done. So, assume that there is a nonzero annihilator $Q$ for $\mathbf f$.
We have $Q(f_1(x),\dots,f_m(x))=0\in\epsilon \aff[\epsilon]$. On the other hand, as $f_i(x)\in \epsilon \aff[\epsilon]$, for all $i\in [m]$; we deduce that $Q(f_1(x),\dots,f_m(x))\bmod \epsilon \aff[\epsilon]$ is $Q(\mathbf 0)$, which is the constant term of $Q$. So it equals zero. By Lemma~\ref{lem_geom}, we have $O\in \overline{\im(f)}$ and again we are done.

Conversely, assume $O\in \overline{\im(f)}$ and we will prove that $x$ exists.
If $O\in \im(f)$, then we can choose $x\in \aff^n$ and we are done. So assume $O\in  \overline{\im(f)}\setminus \im(f)$.
Regard $f$ as a dominant morphism from $\aff^n$ to $\overline{\im(f)}$. Its graph $\Gamma_f$ is cut out in $\aff^n\times \aff^m$ by $Y_i-f_i(X_1,\dots,X_n)$, $i\in[m]$. So $\deg(\Gamma_f)\leq \prod_{i=1}^m \deg(f_i) = D$ by B\'ezout's Theorem.

By Lemma~\ref{lem_aux1}, there exists a curve $C\subseteq \aff^n$ such that $O\in \overline{f(C)}$ and $\deg(C)\leq \deg(\Gamma_f)\leq D$. Pick such a curve $C$. Apply Lemma~\ref{lem_aux2} to $C$, $f|_C$ and $O$, and let $p_1,\dots,p_n\in \epsilon^{-\deg(C)} \aff[[\epsilon]]\subseteq \epsilon^{-D} \aff[[\epsilon]]$ be as given by the lemma.
Then $f_i(p_1,\dots,p_n)\in \epsilon \aff[[\epsilon]]$, for all $i\in [m]$. 

For $i\in [n]$, let $x_i$ be the Laurent polynomial obtained from $p_i$ by truncating the terms of degree greater than $D'$.
When evaluating $f_1,\dots,f_m$, at $(p_1,\dots,p_n)$, such truncation does not affect the coefficient of $\epsilon^k$ for $k\leq 0$ by the choice of $D'$.
So $f_i(x_1,\dots,x_n)\in \epsilon \aff[\epsilon]$,  for all $i\in [m]$.
\end{proof}

\noindent
{\em Remark--} The lower bound $-D=-\prod_{i=1}^m \deg(f_i)$ for the least degree of $x_i$ in $\epsilon$ can be achieved up to a  factor of $1+o(1)$. Consider the polynomials $f_1=f_2=X_1$, $f_3=X_1^{d-1}X_2-1$, and $f_i=X_{i-2}^d-X_{i-1}$ for $i=4,\dots,m$, where $m=n+1$. Then we are forced to choose $x_1\in  \epsilon \aff[\epsilon]$ and $x_i \,\in\, \epsilon^{-(d-1)d^{i-2}}\cdot\aff[\epsilon^{-1}] $, for $i=2,\dots, n$. So the least degree of $x_n$ in $\epsilon$  is at most $-(d-1)d^{n-2}$, while $-D=-d^{n-1}$.

\subsection{Putting APS in PSPACE}
\vspace{-0.05in}

Owing to the exponential upper bound on the precision (= degree wrt $\epsilon$) shown in Theorem \ref{thm_approx}, one expects to solve APS in EXPSPACE only. Surprisingly, in this section, we give a PSPACE algorithm. This we do by reducing the general AnnAtZero instance to a very special instance, that is easy to solve.

Let $\aff$ be the algebraic closure of the field $\F$. Let $f_1,\dots,f_m\in \F[X_1,\dots,X_n]$ be given. 
Denote by $k$ the trdeg of $\F(f_1,\dots,f_m)/\F$. Computing $k$ can be done in PSPACE using linear algebra \cite{Plo05,C76}. 
We assume $k<m-1$, since the cases $k=m-1$ and $k=m$ are again easy to solve in PSPACE using linear algebra.

We reduce the number of polynomials from $m$ to $k+1$ as follows:
Fix a finite subset $S\subseteq \F$,  and choose $c_{i,j}\in S$ at random for $i\in [k+1]$ and $j\in [m]$. For this to work, we need a large enough $S$ and $\F$.
For $i\in [k+1]$, let $g_i:= \sum_{j=1}^m c_{i,j} f_j$. 

Let  $\delta:= (k+1)(\max_{i\in [m]}\deg(f_i))^k/|S|$. Our algorithm is immediate once we prove the following claim.
\begin{theorem}[Random reduction]\label{thm-randReduct}
It holds, with probability $\ge (1-\delta)$, that  

\smallskip\noindent
(1) the transcendence degree of $\F(g_1,\dots,g_{k+1})/\F$ equals $k$, and

\noindent
(2) the constant term of every annihilator for $g_1,\dots,g_{k+1}$ is zero  iff the constant term of  every annihilator for $f_1,\dots,f_{m}$ is zero.
\end{theorem}

First, we reformulate the two items of Theorem~\ref{thm-randReduct} in a geometric way, and later we will analyze the error probability.

For $d\in\N$, denote by $\aff^d$ (resp.~$\proj^d$) the $d$-dimensional affine space (resp.~projective space) over $\aff:=\ol{\F}$.
 Let $f: \aff^n\to \aff^m$ (resp.~$g: \aff^n\to \aff^{k+1}$) be the polynomial map sending $x$ to $(f_1(x),\dots,f_m(x))$ (resp.~$(g_1(x),\dots,g_{k+1}(x))$). Let $O$ and $O'$ be the origin of $\aff^m$ and that of $\aff^{k+1}$ respectively.
Define the affine varieties $V:= \overline{\im(f)}\subseteq \aff^m$ and $V':= \overline{\im(g)}\subseteq \aff^{k+1}$. 
Then $\dim V= \text{trdeg } \mathbf f = k$.

Let $\pi: \aff^m\to \aff^{k+1}$ be the linear map sending $(x_1,\dots,x_m)$ to $(y_1,\dots,y_{k+1})$ where $y_i= \sum_{j=1}^m c_{i,j} x_j$.
Then $g= \pi\circ f$ and $V'= \overline{\pi(V)}$.\footnote{To see $V'\supseteq\overline{\pi(V)}$, note that $\pi^{-1}(V')$ contains $\im(f)$ and is closed, and hence contains $V=\overline{\im(f)}$.}
 Now (1) of  Theorem~\ref{thm-randReduct} is equivalent to $\dim V'=k$, and (2) is equivalent to $O'\in V'$ iff  $O\in V$.
\[
\begin{tikzcd}[column sep=large]
\aff^n \arrow{r}{f} \arrow[swap]{rd}{g}
& V= \overline{\im(f)}  \arrow{r}{\subseteq} \arrow{d}{\pi|_V} & \aff^m \arrow{d}{\pi}\\
& V'= \overline{\im(g)} \arrow{r}{\subseteq} & \aff^{k+1}
\end{tikzcd}
\] 
 We will give sufficient conditions of (1) and (2) in terms of incidence properties.
 Note that $O\in V$ implies $O'\in V'$, since $\pi(O)=O'$. Now suppose $O\not\in V$. Let $W:= \pi^{-1}(O')$, which is a linear  subspace of $\aff^m$.
  Then $O'\not\in \pi(V)$ iff $V\cap W=\emptyset$. However, $V\cap W=\emptyset$ does not imply $O'\not\in V'$, as $V$ may  ``get infinitesimally close to $W$''  without actually meeting $W$, so that $O'\in\overline{\pi(V)}=V'$. See Example \ref{exmp_reduction} in the appendix.
  
  To overcome this problem, we consider projective geometry instead of affine geometry.
  Suppose $\aff^m$ have coordinates $X_1,\dots,X_m$ and $\proj^m$ have homogeneous coordinates $X_0,\ldots,X_m$.
  Regard $\aff^m$ as a dense open subset of $\proj^m$ via $(x_1,\dots,x_m)\mapsto (1,x_1,\dots,x_m)$. Then $H:= \proj^m\setminus \aff^m \,\cong \proj^{m-1}$ is the {\em hyperplane at infinity}, defined by $X_0=0$.
  Denote by $V_c$ (resp.~$W_c$) the {\em projective closure} of $V$ (resp.~$W$) in $\proj^m$. Then $V = V_c\cap \aff^m$.
  Let $W_H:= W_c\cap H$, which is a projective subspace of $H$. 
  
  For distinct points $P,Q\in\proj^m$, write $\overline{PQ}$ for the projective line passing through them. 
\begin{lemma}[Sufficient condns]\label{lem_projcriterion} We have:

\smallskip\noindent
(1) $\dim V'=k$, if ~$V_c\cap W_H=\emptyset$, and

\noindent 
(2) $O'\not\in V'$, if ~$V_c\cap W_c=\emptyset$.
\end{lemma}
\begin{proof}
(1): Assume $\dim V'<k$. Choose $P\in\pi(V)$. The dimension of $\pi^{-1}(P)\cap V$ is at least $\dim V-\dim V' \,\geq 1$ \cite[Thm.11.12]{Har92}.
Denote by $Y$ and $Z$  the projective closure of $\pi^{-1}(P)$ and that of $\pi^{-1}(P)\cap V$ in $\proj^m$ respectively. Then $Z\subseteq Y\cap V_c$. As $\dim Z=\dim \pi^{-1}(P)\cap V\geq 1$ and $\dim H=m-1$, we have $Z\cap H\neq\emptyset$ \cite[Prop.11.4]{Har92}. 

As $\pi$ is a linear map, $\pi^{-1}(P)=Y\cap\aff^m$ is a translate of $\pi^{-1}(O')= W =W_c\cap\aff^m$. It is well known that two projective subspaces $W_1,W_2\not\subseteq H$ have the same intersection with $H$ iff $W_1\cap \aff^m$ and $W_2\cap \aff^m$ are translates of each other.\footnote{Indeed, $W_i\cap \aff^m$ is defined by linear equations $\sum_{j=1}^m a_{j,t} X_j +a_{0,t}=0$ iff $W_i\cap H$ is defined by homogeneous linear equations $X_0=0$ and $\sum_{j=1}^m a_{j,t} X_j=0$. So the constant terms $a_{0,t}$ do not matter.}
So, $Y\cap H=W_c\cap H=W_H$. Therefore,
$V_c\cap W_H \,=\, V_c\cap Y\cap H \,\supseteq\, Z\cap H \,\neq \emptyset
$.

(2): Assume to the contrary that   $V_c\cap W_c=\emptyset$ but $O'\in V'$. We will derive a contradiction. 
As $W_H\subseteq W_c$, we have  $V_c\cap W_H=\emptyset$ and hence $\dim V'=k$ by (1).

Denote by $J(V_c, W_H)$ the {\em join} of $V_c$ and $W_H$, which is defined to be the union of the projective lines $\overline{PQ}$, where $P\in V_c$ and $Q\in W_H$. It is known that $J(V_c, W_H)$, as the join of two {\em disjoint} projective subvarieties, is again a projective subvariety of $\proj^m$ \cite[Example~6.17]{Har92}. Consider $P\in V_c$ and $Q\in W_H$. If $P\in H$, the line $\overline{PQ}$ lies in $H$ and does not meet $\aff^m$. Now suppose $P\in V_c\setminus H=V$. Then $\overline{PQ}$ meets $\overline{OQ}$ at the point $Q$. So $\overline{PQ}\cap \aff^m$ is a translate of $\overline{OQ}\cap \aff^m \,\subseteq\, W_c\cap\aff^m = W$.

Conversely, let $P\in V$. Let $W_P$ denote the unique translate of $W$ containing $P$. Let $\ell_P$ be an affine line contained in $W_P$ and passing through $P$ (note that $W_P$ is the union of such lines). Then $\ell_P$ is a translate of an affine line $\ell\subseteq W$. As $\ell_P$ and $\ell$ are translates of each other, their projective closures intersect $H$ at the same point $Q$. We have $Q\in \ell\cap H \subseteq W_H$. So $\ell_P=\overline{PQ}\cap \aff^m\subseteq J(V_c, W_H)\cap \aff^m$. We conclude that
\begin{equation}\label{eq_join}
J(V_c, W_H)\cap \aff^m \,=\, \bigcup_{P\in V} W_P \,.
\end{equation}
We claim that $J(V_c,W_H)\cap \aff^m=\pi^{-1}(V')$. As $\pi$ is a linear map, Equation \eqref{eq_join} implies $J(V_c,W_H)\cap \aff^m\subseteq\pi^{-1}(V')$. We prove the other direction by comparing dimensions.
It is known that for two {\em disjoint} projective subvarieties $V_1$ and $V_2$,  $\dim J(V_1,V_2)=\dim V_1+\dim V_2+1$ \cite[Prop.11.37-Ex.11.38]{Har92}. Therefore,
\[
\dim J(V_c, W_H)=\dim V_c +\dim W_H+1=\dim V+\dim W=k+\dim W \,.
\]
So, $\dim J(V_c,W_H) \cap \aff^m \,= k+\dim W$.
On the other hand, we have $\pi^{-1}(V')\cong V'\times W$. So $\dim \pi^{-1}(V')=\dim V'+\dim W=k+\dim W$. Now $J(V_c,W_H)\cap \aff^m$ and $\pi^{-1}(V')$ are (irreducible) affine varieties of the same dimension, and one is contained in the other. So they must be equal. This proves the claim.

As $O'\in V'$, we have $W=\pi^{-1}(O')\subseteq \pi^{-1}(V')=\bigcup_{P\in V} W_P$. So $W_P=W$ for some $P\in V$, since $W$ is a linear space. But then $P\in V\cap W_P=V\cap W\subseteq V_c\cap W_c$, contradicting the assumption $V_c\cap W_c=\emptyset$.
 \end{proof}
\noindent {\em Remark-- }
The converse of Lemma~\ref{lem_projcriterion} (Condition 2) is false; see Example \ref{exmp-projcrit} in the appendix.

\smallskip\noindent
{\bf Error probability.} It remains to bound the probability of failure of the conditions $V_c\cap W_H=\emptyset$ and (in the case $O\not\in V$) $V_c\cap W_c=\emptyset$ in Lemma~\ref{lem_projcriterion}. We need the following lemma.

\begin{lemma}[Cut by hyperplanes]\label{lem_hyperplane}
Let $V\subseteq\proj^m$ be a projective subvariety of dimension $r$ and degree $d$. Let $r'\geq r+1$. Choose $c_{i,j}\in S$ at random, for $i\in [r']$ and $0\leq j\leq m$.
Let $W\subseteq \proj^m$ be the projective subspace cut out by the equations $\sum_{j=0}^{m} c_{i,j} X_j=0$, $i=1,\dots,r'$, where $X_0,\dots,X_m$ are homogeneous coordinates of $\proj^m$. Then $V\cap W=\emptyset$ holds with probability at least $1-(r+1)d/|S|$.
\end{lemma}
\begin{proof}
For $i\in [r']$, let $H_i\subseteq \proj^m$ be the hyperplane defined by $\sum_{j=0}^{m} c_{i,j} X_j=0$. By ignoring $H_i$ for $i>r+1$, we may assume $r'=r+1$.  Let $V_0:= V$ and $V_i:= V_{i-1}\cap H_i$ for $i\in [r']$. It suffices to show that $\dim V_i=\dim V_{i-1}-1$ holds with probability at least $1-d/|S|$, for each $i\in [r']$ (the dimension of the empty set is $-1$ by convention). 

Fix $i\in [r']$ and $c_{i',j}$, for $i'\in [i-1]$ and $0\leq j\leq m$. So $V_{i-1}$ is also fixed.  Note that $V_{i-1}\neq \emptyset$ since by taking a hyperplane section reduces the dimension  by at most one. If $\dim V_i\neq \dim V_{i-1}-1$, then $\dim V_i=\dim V_{i-1}$, and $H_i$ contains some irreducible component of $V_{i-1}$ \cite[Exercise~11.6]{Har92}. Let $Y$ be an irreducible component of $V_{i-1}$, and fix a point $P\in Y$. Then $Y\subseteq H_i$ only if $P\in H_i$, which holds only if  $c_{i,0},\dots,c_{i,m}$ satisfy a nonzero linear equation determined by $P$. This occurs with probability at most $1/|S|$ (eg.~by fixing all but one $c_{i,j}$). 
We also have $\deg(V_{i-1})\leq \deg(V)\leq d$, and hence the number of irreducible components of $V_{i-1}$ is bounded by $d$.
By the union bound, $H_i$ contains an irreducible component of $V_{i-1}$ with probability at most $d/|S|$.
\end{proof}
\begin{proof}[Proof of Theorem~\ref{thm-randReduct}]
As mentioned above, Theorem~\ref{thm-randReduct} is equivalent to showing that, with probability at least $1-\delta$:  (1) $\dim V'=k$, and (2) $O'\in V'$ iff  $O\in V$.
Note that $W_c$ is cut out in $\proj^m$ by the linear equations $\sum_{j=1}^{m} c_{i,j} X_j=0$, $i=1,\dots,k+1$.
So $W_H$ is cut out in $H\cong \proj^{m-1}$ (corresponding to $X_0=0$) by the linear equations $\sum_{j=1}^m c_{i,j} X_j=0$, $i=1,\dots,k+1$. 
We also have $\deg(V_c\cap H)\leq \deg(V_c)\leq (\max_{i\in [m]}\deg(f_i))^k$ (see, e.g., \cite[Thm.8.48]{burgisser2013algebraic}).

Assume $O\in V$. Then $O'\in V'$ since $\pi(O)=O'$. 
Applying Lemma~\ref{lem_hyperplane} to  each of the irreducible components of $V_c\cap H$ and $W_H$, as subvarieties of $H\cong\proj^{m-1}$, we see $V_c\cap W_H=(V_c\cap H)\cap W_H=\emptyset$ holds with probability at least $1-k\deg(V_c \cap H)/|S|\geq 1-\delta$.
So by Lemma~\ref{lem_projcriterion}, $\dim V'=k$ holds with probability at least $1-\delta$.

Now assume $O\not\in V$. 
Let $\pi_{O, H}: V_c\to H$ be the {\em projection of $V_c$ from $O$ to $H$}, defined by $P\mapsto \overline{OP}\cap H$ for $P\in V_c$. It is well defined since $O\not\in V_c$. 
The image $\pi_{O,H}(V_c)$ is a projective subvariety of $H$ \cite[Thm.3.5]{Har92}. 
If $V_c\cap W_c$ contains a point $P$, then $\pi_{O,H}(V_c)\cap W_H$ contains $\pi_{O,H}(P)$.
Conversely, if $\pi_{O,H}(V_c)\cap W_H$ contains a point $Q$, then there exists $P\in V_c$ such that $Q=\pi_{O,H}(P)$, and we have $P\in\overline{O Q}\subseteq W_c$.
We conclude that $\pi_{O,H}(V_c)\cap W_H=\emptyset$ iff  $V_c\cap W_c=\emptyset$, which implies  $V_c\cap W_H=\emptyset$. 

Note that $\dim \pi_{O,H}(V_c)=\dim V_c=k$, since  $\pi_{O,H}(V_c)=J(\{O\}, V_c)\cap H$.
We also have $\deg(\pi_{O,H}(V_c))\leq \deg(V_c)$ \cite[Eg.18.16]{Har92}.
Applying Lemma~\ref{lem_hyperplane} to $\pi_{O,H}(V_c)$ and $W_H$, as subvarieties of $H\cong\proj^{m-1}$, we see $\pi_{O,H}(V_c)\cap W_H=\emptyset$ holds with probability at least $1-(k+1)\deg(\pi_{O,H}(V_c))/|S|\geq 1-\delta$.

By  Lemma~\ref{lem_projcriterion} and the previous paragraphs, it holds with probability at least $1-\delta$ that $\dim V'=k$ and $O'\not\in V'$.
\end{proof}
\begin{proof}[Proof of Theorem \ref{thm-aps}]
AnnAtZero is known to be NP-hard \cite{Kay09}. The NP-hardness of APS follows from Lemma \ref{lem_geom} and Theorem \ref{thm_approx}.

Given an instance $\mathbf f$ of APS, we can first find the trdeg $k$. Fix a subset 
$S\subset \aff$ to be larger than $2(k+1)(\max_{i\in [m]}\deg(f_i))^k$ (which can be scanned using only polynomial-space). Consider the points $\left(\left(c_{i,j} \,\mid\, i\in [k+1],\, j\in [m] \right)\right)\in S^{(k+1)\times m}$; for each such point define $\mathbf g:= \big\{g_i:= \sum_{j=1}^m c_{i,j} f_j \,\mid$ $i\in[k+1] \big\}$.   
Compute the trdeg of $\mathbf g$, and if it is $k$ then solve AnnAtZero for the instance $\mathbf g$. 
Output NO iff some $\mathbf g$ failed the AnnAtZero test.

All these steps can be achieved in space polynomial in the input size, using the uniqueness of the annihilator for $\mathbf g$ \cite[Lem.7]{Kay09}, Perron's degree bound \cite{Plo05} and linear algebra \cite{C76}.	
\end{proof}

\section{Hitting-set for $\overline{\rm VP}$: Proof of Theorem \ref{thm-hsg}} \label{sec-hsg}
\vspace{-0.05in}

Suppose $p$ is a prime. Define $\aff:=\ol{\F}_p$. We want to find hitting-sets for certain polynomials in $\aff[x_1,\dots,x_n]$. Fix a $p$-power $q\ge\Omega(sr^6)$, for the given parameters $s, r$. Assume that $p\nmid(r+1)$. Also, fix a model for the finite field $\F_q$ \cite{AL86}. We now define the notion of `infinitesimally approximating' a polynomial by a small circuit.

\smallskip\noindent
{\bf Approximative closure of VP.} \cite{bringmann2017algebraic}
  A family $(f_n|n)$ of polynomials from $\aff[\mathbf x]$ is in the {\em class $\overline{\rm VP}_\aff$} if there are polynomials $f_{n,i}$ and a function $t:\mathbb{N} \mapsto \mathbb{N}$ such that $g_n$  has a poly($n$)-size poly($n$)-degree algebraic circuit, over the field $\aff(\epsilon)$, computing $g_n(\mathbf x)= f_n(\mathbf x)+ \epsilon f_{n,1}(\mathbf x)+ {\epsilon}^2 f_{n,2}(\mathbf x)+ \ldots+ {\epsilon}^{t(n)} f_{n,t(n)}(\mathbf x)$. That is, $g_n \equiv f_n \bmod{\epsilon\aff[\epsilon][\mathbf x]}$.
  
The smallest possible circuit size of $g_n$ is called the {\em approximative complexity} of $f_n$, namely $\ol{\sz}(f_n)$. 

It may happen that $g_n$ is much easier than $f_n$ in terms of traditional circuit complexity. That possibility makes the definition interesting and opens up a long line of research. 

\smallskip\noindent
{\bf Hitting-set for $\overline{\rm VP}_\aff$.}
Given functions $s=s(n)$ and $r=r(n)$, a finite subset $\calH\subset \aff^n$ is called a {\em hitting-set} for degree-$r$ polynomials of approximative complexity $s$, if for every such nonzero polynomial $f$: $\exists \mathbf v\in\calH, \, f(\mathbf v)\ne0$.

\smallskip
{\bf Explicitness.}
We are interested in computing such a hitting-set in poly($s, \log r, \log q$)-time. 

Before our work, the best result known was EXPSPACE \cite{mul12, mulmuley2017geometric}. Heintz and Schnorr \cite{heintz1980testing} proved that poly$(s,\log qr)$-sized hitting-sets exist aplenty (for degree-$r$ $\ol{\sz}$-$s$ polynomials).

\begin{lemma}{\em \cite[Thm.4.4]{heintz1980testing}}\label{lem-hs80}
 There exists a hitting-set $\mathcal{H}\subset \F_q^n $ of size $O(s^2n^2)$ (assuming $q\geq\Omega(sr^2)$) that hits all nonzero degree-$r$ $n$-variate polynomials in $\aff[\mathbf x]$ that can be infinitesimally approximated by size-$s$ algebraic circuits. 
\end{lemma}

Note that for the hitting-set design problem it suffices to focus only on homogeneous polynomials. They are known to be computable by homogeneous circuits, where each gate computes a homogeneous polynomial (see \cite{shpilka2010arithmetic}).

\smallskip\noindent
{\bf Universal circuit.}
It can simulate any circuit of size-$s$ computing a degree-$r$ homogeneous polynomial in $\aff(\epsilon)[x_1,\ldots,x_n]$. We define the {\em universal circuit} $\Psi(\mathbf{y},\mathbf{x})$ as a circuit in $n$ essential variables $\mathbf{x}$ and $s':=O(sr^4)$ auxiliary variables $\mathbf y$. The variables $\mathbf y$ are the ones that one can specialize in $\aff(\epsilon)$, to compute a specific polynomial in $\aff(\epsilon)[x_1,\ldots,x_n]$. Every specialization gives a homogeneous degree-$r$ $\ol{\sz}$-$s'$ polynomial. Moreover, the set of these polynomials is closed under constant multiples (see \cite[Thm.2.2]{forbes2017pspace}).

Note that by \cite{heintz1980testing} there is a hitting-set, with $m:=O(s'^2n^2)$ points in $\F_q^n$ ($\because q\ge\Omega(s'r^2)$), for the set of polynomials $\calP$ approximated by the specializations of $\Psi(\mathbf{y},\mathbf{x})$.
A universal circuit construction can be found in \cite{raz2008elusive, shpilka2010arithmetic}. 
Using the above notation, we give a criterion to decide whether a candidate set is a hitting-set. 

\begin{theorem}[hs criterion]\label{thm_HITTING}
Set $\calH=:\{\mathbf v_1,\ldots, \mathbf v_m\}$ $\subset \F_q^n$ is {\em not} a hitting-set for the family of polynomials $\calP$ iff there is a satisfying assignment $(\alpha, \beta)\,\in \aff(\epsilon)^{s'}\times \aff(\epsilon)^{n}$ such that:

\smallskip\noindent
(1) $\forall i \in [n],\, {\beta_i}^{r+1}- 1 \,\in \epsilon\aff[\epsilon]$, and

\noindent 
(2) $\Psi(\alpha, \beta) - 1 \,\in \epsilon\aff[\epsilon]$, and

\noindent
(3) $\forall i \in [m],\, \Psi(\alpha,\mathbf{v}_i) \,\in \epsilon\aff[\epsilon]$.
\end{theorem}
{\bf Remark--} The above criterion holds for algebraically closed fields $\aff$ of {\em any} characteristic. Thus, it reduces those hitting-set design problems to APS as well.
\begin{proof}
First we show that: $\exists x\in \aff(\epsilon),\, {x}^{r+1}- 1 \,\in \epsilon\aff[\epsilon]$ implies $x\in \aff[[\epsilon]]\cap \aff(\epsilon)$ (= rational functions defined at $\epsilon=0$). 

Recall the formal power series $\aff[[\epsilon]]$ and its group of units $\aff[[\epsilon]]^*$.
Note that for any polynomial $a=\big(\sum_{i_0\le i\le d} a_i\epsilon^i\big)$ with $a_{i_0}\ne0$, the inverse $a^{-1} = \epsilon^{-i_0}\cdot \big(\sum_{i_0\le i\le d} a_i\epsilon^{i-i_0}\big)^{-1}$ is in $\epsilon^{-i_0}\cdot \aff[[\epsilon]]^*$. This is just a consequence of the identity $(1-\epsilon)^{-1} = \sum_{i\ge0} \epsilon^i$. In other words, any rational function $a\in \aff(\epsilon)$ can be written as an element in $\epsilon^{-i}\aff[[\epsilon]]^*$, for some $i\ge0$. Thus, write $x$ as $\epsilon^{-i}\cdot(b_0+b_1\epsilon+\cdots)$ for $i\ge0$ and $b_0\in\aff^*$. This gives
$${x}^{r+1}- 1 \,=\, \epsilon^{-i(r+1)}(b_0+b_1\epsilon+ b_2\epsilon^{2} +\cdots)^{r+1} \;-\; 1 \,.$$
For this to be in $\epsilon\aff[\epsilon]$, clearly $i$ has to be $0$ (otherwise, $\epsilon^{-i(r+1)}$ remains uncancelled); implying that $x\in \aff[[\epsilon]]$. 

Moreover, we deduce that $b_0^{r+1} - 1 = 0$. Thus, condition (1) implies that $b_0$ is one of the $(r+1)$-th roots of unity $Z_{r+1}\subset\aff$ (recall that, since $p\nmid(r+1)$, $|Z_{r+1}|=r+1$). Thus, $x\in \, Z_{r+1} + \epsilon\aff[[\epsilon]]$.

\smallskip{\em [$\Rightarrow$]:}
Suppose $\mathcal{H}$ is not a hitting-set for $\calP$. Then, there is a specialization $\alpha \in \aff(\epsilon)^{s'}$ of the universal circuit such that $\Psi(\alpha, \mathbf x)$ computes a polynomial in $\aff[\epsilon][\mathbf x]\setminus \epsilon\aff[\epsilon][\mathbf x]$, but still `fools' $\calH$, i.e.:
$\forall i \in [m],\, \Psi(\alpha,\mathbf{v}_i) \,\in \epsilon\aff[\epsilon]$. 
What remains to show is that conditions (1) and (2) can be satisfied too.

Consider the polynomial $g(\mathbf x):= \Psi(\alpha, \mathbf x)\big\vert_{\epsilon=0}$. It is a nonzero polynomial, in $\aff[\mathbf x]$ of degree-$r$, that `fools' $\calH$. By \cite{Sch80}, there is a $\beta\in Z_{r+1}^n$ such that $a:= g(\beta)$ is in $\aff^*$. Clearly, $\beta_i^{r+1} - 1 = 0$, for all $i$. Consider $\psi':= a^{-1}\cdot\Psi(\alpha, \mathbf x)$. Note that $\psi'(\beta)-1 \in \epsilon\aff[\epsilon]$, and $\psi'(\mathbf{v}_i) \,\in \epsilon\aff[\epsilon]$ for all $i$. Moreover, the normalized polynomial $\psi'(\mathbf x)$ can easily be obtained from the universal circuit $\Psi$ by changing one of the coordinates of $\alpha$ (eg.~the incoming wires of the root of the circuit). This means that the three conditions (1)-(3) can be simultaneously satisfied by (some) $(\alpha',\beta)\in\, \aff(\epsilon)^{s'}\times Z_{r+1}^n$.

\smallskip{\em [$\Leftarrow$]:}
Suppose the satisfying assignment is $(\alpha, \beta')\,\in \aff(\epsilon)^{s'}\times \aff(\epsilon)^{n}$. As shown before, condition (1) implies:  $\beta'_i\in \, Z_{r+1} + \epsilon\aff[[\epsilon]]$ for all $i\in[n]$. Let us define $\beta_i:= \beta'_i\big\vert_{\epsilon=0}$, for all $i\in[n]$; they are in $Z_{r+1}\subset\aff$. 
By Condition (3): $\forall i \in [m],\, \Psi(\alpha,\mathbf{v}_i) \,\in \epsilon\aff[\epsilon]$.

Previous calculations suggest that $\Psi(\alpha,\mathbf x)$ is in 
$\epsilon^{-j}\aff[[\epsilon]][\mathbf x]$, for some $j\ge0$.
Expand the polynomial $\Psi(\alpha,\mathbf x)$, wrt $\epsilon$, as: 
$$g_{-j}(\mathbf x)\epsilon^{-j} + \dots + \epsilon^{-2} g_{-2}(\mathbf x) + g_{-1}(\mathbf x)\epsilon^{-1} + g_0(\mathbf x) + \epsilon g_1(\mathbf x) + \epsilon^2 g_2(\mathbf x) + \dots \,.$$
Let us study Condition (2). If for each $0\le\ell\le j$, polynomial $g_{-\ell}(\mathbf x)$ is zero, then $\Psi(\alpha,\beta')\big\vert_{\epsilon=0} =0$ contradicting the condition. Thus, we can pick the largest $0\le\ell\le j$ such that the polynomial $g_{-\ell}(\mathbf x)\ne0$. 

Note that the normalized circuit $\epsilon^\ell\cdot \Psi(\alpha,\mathbf x)$ equals $g_{-\ell}$ at $\epsilon=0$. This means that $g_{-\ell}\in\calP$, and it is a nonzero polynomial fooling $\calH$. Thus, $\calH$ cannot be a hitting-set for $\calP$ and we are done.
\end{proof}

\begin{proof}[Proof of Theorem \ref{thm-hsg}]
Given a prime $p$ and parameters $n, r, s$ in unary (wlog $p\nmid(r+1)$), fix a field $\F_q$ with $q\ge\Omega(sr^6)$. Fix the universal circuit $\Psi(\mathbf{y},\mathbf{x})$ with $n$ essential variables $\mathbf{x}$ and $s':=\Omega(sr^4)$ auxiliary variables $\mathbf y$. Fix $m:=\Omega(s'^2n^2)$.

For every subset $\calH=:\{\mathbf v_1,\ldots, \mathbf v_m\}$ $\subset \F_q^n$ solve the APS instance described by Conditions (1)-(3) in Theorem \ref{thm_HITTING}. These are $(n+m+1)$ algebraic circuits of degree poly($srn, \log p$) and a similar bitsize. Using the algorithm from Theorem \ref{thm-aps} it can be solved in poly($srn, \log p$)-space.

The number of subsets $\calH$ is $q^{nm}$. So, in poly($nm\log q$)-space we can go over all of them. If APS fails on one of them (say $\calH$) then we know that $\calH$ is a hitting-set for $\calP$. Since $\Psi$ is universal, for homogeneous degree-$r$ $\ol{\sz}$-$s$ polynomials in $\aff[\mathbf x]$, we output $\calH$ as the desired hitting-set.
\end{proof}

\vspace{-0.1in}
\section{Conclusion}\label{sec-conclusion}
\vspace{-0.1in} 

Our result on algebraic dependence testing in AM $\cap$ coAM gives further indication that a randomized polynomial time algorithm for the problem exists. Studying the following special case might be helpful to get an idea for designing better algorithms.

  Given quadratic polynomials $f_1,\ldots,f_n \in \F_2[x_1,\ldots,x_n]$, test if they are algebraically dependent in randomized polynomial time \cite{pandey2016algebraic}.

\smallskip
 As indicated in this paper, approximate polynomials satisfiability, or equivalently testing zero-membership in the Zariski closure of the image, may have further applications to problems in computational algebraic geometry and algebraic complexity. 

We know that HN is in AM over characteristic zero fields, assuming GRH \cite{koiran1996hilbert}. Can we solve  AnnAtZero (or APS) in AM for characteristic zero fields assuming GRH?) \cite{Kay09}? This would also imply better hitting-set construction for $\ol{\rm VP}$.

\bigskip\noindent
{\bf Acknowledgements.}
We thank Anurag Pandey and Sumanta Ghosh for insightful discussions on the approximate polynomials satisfiability and the hitting-set construction problems. 
N.S.~thanks the funding support from DST (DST/SJF/MSA-01/2013-14).
Z.G.~is funded by DST and Research I Foundation of CSE, IITK.

\bibliographystyle{alpha} 
\bibliography{reference}

\appendix

\section{From Section \ref{sec-prelim}: Algebraic-Geometry  }\label{app-AG}

Let $\aff:=\ol{\F}$ be the algebraic closure of a field $\F$. For $d\in\N^+$, write $\aff^d$ for the {\em $d$-dimensional affine space} over $\aff$. It is defined to be the set $\aff^d$, equipped with the {\em Zariski topology}, defined as follows: A subset $S$ of $\aff^d$ is {\em closed} iff it is the set of common zeros of some subset of polynomials in $\aff[X_1,\dots,X_d]$. For other subsets $S$ it makes sense to consider the {\em closure} $\ol{S}$-- the smallest closed set containing $S$. Set $S$ is {\em dense} if $\ol{S}=\aff^d$.  Complement of closed sets are called {\em open}. 

A closed set is called a {\em hypersurface} (resp.~{\em hyperplane}) if it is definable by a single polynomial (resp.~single linear polynomial). 

Define $\aff^\times:= \aff\setminus\{0\}$.
Write $\proj^d$ for the $d$-dimensional projective space over $\aff$, defined to be the quotient set $(\aff^{d+1}\setminus\{(0,\dots,0)\})/\sim$. Where $(x_0,\dots,x_d)\sim (y_0,\dots,y_d)$ iff there exists $c\in\aff^\times$ such that $y_i=c x_i$ for $0\leq i\leq d$. The set $\proj^d$ is again equipped with the {\em Zariski topology}, where a subset is closed iff it is the set of common zeros of some subset of {\em homogeneous} polynomials in  $\aff[X_0,\dots,X_d]$. We use $(d+1)$-tuples $(x_0,\dots,x_d)$ to represent points in $\proj^d$.

\smallskip
Closed subsets of $\aff^d$ or $\proj^d$ are also called {\em algebraic sets} or {\em zerosets}. 
An algebraic set is {\em irreducible} if it cannot be written as the union of finitely many proper algebraic sets. An irreducible algebraic subset of an affine (resp. projective) space is also called an {\em affine variety} (resp. {\em projective variety}). (In some references, varieties are not required to be irreducible, but in this work we always assume it.)
An algebraic set $V$ can be uniquely represented as the union of finitely many varieties, and these varieties are called the {\em irreducible components} of $V$.

Affine zerosets (resp.~varieties) are in 1-1 correspondence with {\em radical} (resp.~{\em prime}) ideals.
Irreducible decomposition of an affine variety mirrors the factoring of an ideal into primary ideals.  Finally, note that the affine points are in 1-1 correspondence with {\em maximal} ideals; it is a simple reformulation of Hilbert's Nullstellensatz.

\smallskip
The affine space $\aff^d$ may be regarded as a subset of $\proj^d$ via the map $(x_1,\dots,x_d)\mapsto (1,x_1,\dots,x_d)$. 
Then the subspace topology of $\aff^d$ induced from the Zariski topology of $\proj^d$ is just the Zariski topology of $\aff^d$.
The set $\proj^d\setminus \aff^d$ is the projective subspace of $\proj^d$ defined by $X_0=0$, called the {\em hyperplane at infinity}.
 
For an algebraic subset $V$ of $\aff^d\subseteq\proj^d$, the smallest algebraic subset $V'$ of $\proj^d$ containing $V$ (i.e. the intersection of all algebraic subsets containing $V$) is the {\em projective closure} of $V$, and we have $V'\cap \aff^d=V$. To see this, note that for $P=(x_1,\dots,x_d)\in \aff^d\setminus V$, there exists a polynomial $Q\in\aff[X_1,\dots,X_d]$ of degree $D\in\N$ not vanishing on $P$ (but vanishing on $V$). Then its homogenization $Q'\in\aff[X_0,\dots,X_d]$, defined by replacing each monomial $M=\prod_{i=1}^d X_i^{d_i}$ by  $X_0^{D-\deg(M)}\prod_{i=1}^d X_i^{d_i}$, does not vanish on $(1,x_1,\dots,x_d)$. So, $(1,\mathbf x)\notin V'$.

\smallskip
For distinct points $P=(x_0,\dots,x_d), Q=(y_0,\dots,y_d)\in\proj^d$, write $\overline{PQ}$ for the {\em projective line} passing through them, i.e., $\overline{PQ}$ consists of the points $(ux_0+vy_0,\dots, ux_d+vy_d)$, where $(u,v)\in\aff^2\setminus\{(0,0)\}$. 

\smallskip
The {\em dimension} of a variety $V$ is defined to be the largest integer $m$ such that there exists a chain of varieties $\emptyset\subsetneq V_0\subsetneq V_1\subsetneq\cdots\subsetneq V_m=V$. More generally, the dimension of an algebraic set $V$, denoted by $\dim V$, is the maximal dimension of its irreducible components. Eg.~we have $\dim\aff^d=\dim\proj^d=d$. 
The dimension of  the  empty set is $-1$ by convention. One dimensional varieties are called {\em curves}.

\smallskip
The {\em degree} of a variety $V$ in $\aff^d$ (resp.~$\proj^d$) is the number of intersections of $V$ with a general affine subspace (resp.~projective subspace) of dimension $d-\dim V$. 
More generally, the degree of an algebraic set $V$, denoted by $\deg(V)$,  is the sum of the degrees of its irreducible components.
The degree of an algebraic subset of $\aff^d$ coincides with the degree of its projective closure in $\proj^d$.

\smallskip
Suppose $V\subseteq\aff^d$ is an algebraic set, defined by polynomials $f_1,\dots,f_k$. Let $(a_1,\dots,a_d)\in\aff^d$. Then the set $\{(x_1+a_1,\dots,x_d+a_d): (x_1,\dots,x_d)\in V\}$ is called a {\em translate} of $V$. It is also an algebraic set, defined by $f_i(X_1-a_1,\dots,X_d-a_d)$, $i=1,\dots,k$.

\smallskip
Let $V\subseteq \aff^n$, $W\subseteq\aff^m$ be affine varieties. A {\em morphism} from $V$ to $W$ is a function $f:V\to W$ that is a restriction of a polynomial map $\aff^n\to\aff^m$.
A morphism $f: V\to W$ is called {\em dominant} if $\overline{\im(f)}=W$.
The preimage of a closed subset under a morphism is closed (i.e. morphisms are {\em continuous} in the Zariski topology).

For a polynomial map $f:\aff^n\to\aff^m$ and an affine variety $V\subseteq \aff^n$, $W:=\overline{f(V)}$ is also an affine variety (i.e., it is irreducible). To see this, assume to the contrary that $W$ is the union of two proper closed subsets $W_1$ and $W_2$. By the definition of closure, $f(V)$ is not contained in either $W_1$ or $W_2$, i.e., it intersects both. Then $f^{-1}(W_1)\cap V$ and $f^{-1}(W_2)\cap V$ are two proper closed subsets of $V$, and their union is $V$. This contradicts the irreducibility of $V$.

\smallskip
The {\em graph} $\Gamma_f$ of a morphism $f$ is the set $\{(x,f(x)): x\in V\}\subseteq V\times W\subseteq \aff^n\times \aff^m$.
Here $V\times W=\{(x,y): x\in V,y\in W\}$ denotes the {\em product} of $V$ and $W$, which is a subvariety of the $(n+m)$-dimensional affine space $\aff^n\times \aff^m\cong \aff^{n+m}$.
Note the graph $\Gamma_f$ is closed in $\aff^n\times \aff^m$: 
Suppose $f$ sends $x\in V$ to $(f_1(x),\dots,f_m(x))\in\aff^m$, where $f_i\in\aff[X_1,\dots,X_n]$ for $i\in [m]$.
And suppose $V$ and $W$ are defined by ideals $I\subseteq \aff[X_1,\dots,X_n]$ and $I'\subseteq \aff[Y_1,\dots,Y_m]$ respectively. Then $\Gamma_f$ is defined by $I$, $I'$, and the polynomials $Y_i-f_i(X_1,\dots,X_n)\in \aff[X_1,\dots,X_n,Y_1,\dots,Y_m]$, $i=1,\dots,m$.

\section{From Section \ref{sec:aps}}

  \begin{exmp}\label{exmp_reduction}
Let $m=4$, $(f_1,f_2,f_3,f_4)=(X_1,X_2,X_1 X_2-1, X_1+X_2)$. Then $k:=\text{trdeg}\mathbf f= 2$.
  Let $(g_1,g_2,g_3)=(f_1,f_3,f_1+f_2-f_4)=(X_1,X_1 X_2-1,0)$.
  Suppose $\aff^m$ has coordinates $Y_1,\dots,Y_4$ and $\aff^{k+1}$ has coordinates $Z_1,\dots,Z_3$.

Then $V\subseteq\aff^m$ is defined by $Y_1Y_2-Y_3-1=0$ and $Y_1+Y_2-Y_4=0$, and $W$ is defined by $Y_1=0$, $Y_3=0$, and $Y_2-Y_4=0$.
So $V\cap W=\emptyset$.  But $V'\subseteq\aff^{k+1}$ is the plane $Z_3=0$, which contains the origin.
  \end{exmp}

\begin{exmp}\label{exmp-projcrit}
Consider Example~\ref{exmp_reduction} but choose $f_4$ to be $X_1+X_2+1$ instead of $X_1+X_2$. Now we have $g_3=1$, $V$ is defined by $Y_1Y_2-Y_3-1=0$ and $Y_1+Y_2-Y_4+1=0$, and $V'$ is the plane $Z_3=1$. So $O'\not\in V'$. 

On the other hand, suppose $\proj^m$ has coordinates $Y_0,\dots,Y_4$. Then $V_c\cap H$ is defined by $Y_0=Y_1Y_2=Y_1+Y_2-Y_4=0$, and $W_H$ is defined by $Y_0=Y_1=Y_2-Y_4=Y_3=0$. So $(0,0,1,0,1)\in V_c\cap W_H\subseteq V_c\cap W_c$.
\end{exmp}

\end{document}